\documentclass[11pt]{amsart}
\usepackage[letterpaper, margin=1in]{geometry}			
\usepackage{graphicx, color, amssymb, url, hyperref, tikz, enumerate} 

\newtheorem{thm}{Theorem}[section]
\newtheorem{lemma}[thm]{Lemma}
\newtheorem{prop}[thm]{Proposition}

\newtheorem{lemma*}{Lemma}

\theoremstyle{definition}
\newtheorem{remark}[thm]{Remark}
\newtheorem{example}[thm]{Example}

\newtheorem{definition}[thm]{Definition}

\newcommand{\supp}{\operatorname{supp}}
\newcommand{\CF}{\mathrm{CF}}
\newcommand{\LT}{\mathrm{LT}}

\newcommand{\psm}[2]{x_{#1}\prod_{i\in #2}(1+x_i)}
\newcommand{\psmj}[2]{x_{#1}\prod_{j\in #2}(1+x_j)}

\newcommand{\union}{\bigcup}
\newcommand{\intersect}{\bigcap}

\newcommand{\biring}{\mathbb{F}_2[x_1,\dots,x_n]}
\newcommand{\dfn}[1]{\textbf{#1}}
\newcommand{\bifield}{\{0,1\}^n}
\newcommand{\h}{-1.3}
\newcommand{\gb}{Gr\"{o}bner basis}
\newcommand{\gbs}{Gr\"{o}bner bases}

\title{Gr\"obner Bases of Neural Ideals}

\author[Garcia]{Rebecca Garcia}
\author[Garc\'{\i}a Puente]{Luis David Garc\'{\i}a Puente}
\address[R. Garcia and L.~D. Garc\'{\i}a Puente]{Department of Mathematics and Statistics,
Sam Houston State University,
Huntsville, TX 77341-2206}
\email{rgarcia@shsu.edu}
\email{lgarcia@shsu.edu}

\author[Kruse]{Ryan Kruse}
\address[R. Kruse]{Mathematics Department, Central College,  Pella, IA 50219}
\email{kruser1@central.edu}

\author[Liu]{Jessica Liu}
\address[J. Liu]{Department of Mathematics, Bard College,
Annandale, NY 12504}
\email{yl7847@bard.edu}

\author[Miyata]{Dane Miyata}
\address[D. Miyata]{Department of Mathematics,
Willamette University,
Salem, OR 97301}
\email{dmiyata@willamette.edu}

\author[Petersen]{Ethan Petersen}
\address[E. Petersen]{Department of Mathematics, Rose-Hulman Institute of Technology, Terre Haute, IN 47803}
\email{peterseo@rose-hulman.edu}

\author[Phillipson]{Kaitlyn Phillipson}
\address[K. Phillipson]{Department of Mathematics, St. Edward’s University,
Austin, Texas 78704-6489}
\email{kphillip@stedwards.edu}

\author[Shiu]{Anne Shiu}
\address[A. Shiu]{Department of Mathematics, Texas A\&M University, College Station, TX 77843}
\email{annejls@math.tamu.edu}

\date{\today}

\begin{document}

\begin{abstract}
The brain processes information about the environment via neural codes. 
The neural ideal was introduced recently as an algebraic object that 
can be used to better understand the combinatorial structure of neural
codes.
Every neural ideal has a particular generating set, 
called the canonical form, that directly encodes 
a minimal description of the receptive field structure intrinsic to the neural code. 
On the other hand, for a given monomial order, 
any polynomial ideal is also generated by its unique (reduced) Gr\"obner basis with respect to 
that monomial order.   
How are these two types of generating sets -- canonical forms and Gr\"obner bases -- related?
Our main result states that if the canonical form of a neural ideal is a Gr\"obner basis, then it is the universal Gr\"obner basis (that is, the union of all reduced Gr\"obner bases).  
Furthermore, we prove that this situation -- when the canonical form is a Gr\"obner basis -- occurs precisely when 
the universal Gr\"obner basis contains only pseudo-monomials (certain generalizations of monomials).
Our results motivate two questions: (1)~When is the canonical form a Gr\"obner basis?
(2)~When the universal Gr\"obner basis of a neural ideal is {\em not} a canonical form, what can the non-pseudo-monomial elements in the basis tell us about the 
receptive fields of the code?
We give partial answers to both questions.  Along the way, we develop a representation of pseudo-monomials as hypercubes in a Boolean lattice.

\smallskip
\noindent {\bf Keywords:} neural code, receptive field, canonical
form, Gr\"obner basis, Boolean lattice

\smallskip
\noindent {\bf MSC classes:} 92-04 (Primary), 13P25, 68W30 (Secondary)
\end{abstract}


\maketitle

\section{Introduction}\label{sec:intro}

The brain is tasked with many important functions, but one of the least understood is how it builds an understanding of the world. 
Stimuli in one's environment are not experienced in isolation, but in relation to other stimuli.
How does the brain represent this organization?  
Or, 
to quote from Curto, Itskov, Veliz-Cuba, and Youngs, 
``What can be inferred about the underlying stimulus space from neural activity alone?'' \cite{neural_ring}. 

Curto {\em et al.}\ pursued this question for codes where each neuron has a region 
of stimulus space, called its {\em receptive field}, in which it fires at a high rate.
They introduced algebraic objects that summarize neural-activity data, which are in the form of {\em neural codes} ($0/1$-vectors where $1$ means the corresponding neuron is active, and $0$ means silence)~\cite{neural_ring}.
The {\em neural ideal} of a neural code is an ideal 
that contains the full combinatorial data of the code.
The {\em canonical form} of a neural ideal is a generating set that 
is a minimal description of the receptive-field structure.  Hence, the questions posed above 
have been investigated via the neural ideal or the canonical form~\cite{ what-makes,neural_ring,neural-hom,new-alg}.  As a complement to algebraic approaches, combinatorial and topological arguments are employed in related works~\cite{intersection-complete,sparse,LSW}.


The aim of our work is to investigate, for the first time, how the canonical form is related to
other generating sets of the neural ideal, namely, its Gr\"obner bases.
This is a natural mathematical question, and additionally the answer could improve algorithms for computing the canonical form.
Currently, there are two distinct methods to compute the canonical form of a neural ideal: the original method proposed in \cite{neural_ring} and an iterative method introduced in \cite{neural-ideal-sage}. The former method requires the computation of primary decomposition of pseudo-monomial ideals. As a result, this method is rather inefficient. Even in dimension 5, one  can find codes for which this algorithm takes hundreds or even thousands of seconds to terminate or halts due to lack of memory.
The more recent iterative method relies entirely on basic polynomial arithmetic. This algorithm can efficiently compute canonical forms for codes in up to 10 dimensions; see \cite{neural-ideal-sage}. On the other hand, Gr\"obner basis computations are generally computationally expensive. Nevertheless, we take full advantage of tailored methods for Gr\"obner basis over Boolean rings \cite{polybori}. As we show later in Table \ref{runtimes}, for small dimensions less than or equal to 8, Gr\"obner basis computations are faster than canonical form ones. For larger dimensions, we have observed that in general Gr\"obner basis computations are faster but the standard deviation on computational time is much larger. In dimension 9, the average time to compute a Gr\"obner basis is around 3 seconds, but there are codes for which that computation takes  close to 10 hours to finish. 

Nevertheless, we believe that a thorough study of Gr\"obner basis of neural ideals is not only of theoretical interest, but it can lead to better procedures able to perform computations in larger dimensions. 
Indeed, among small codes, surprisingly many have canonical forms that are also Gr\"obner bases. Moreover, the iterative nature of the newer canonical form algorithm hints towards 
the ability to compute canonical forms and Gr\"obner bases of neural codes in large dimensions by `gluing' those of codes on small dimensions. Such decomposition results are a common theme in other areas of applied algebraic geometry 
\cite{AR08, EKS14}.



The outline of this paper is as follows.  Section~\ref{sec:background} provides background on neural ideals, canonical forms, and Gr\"obner bases.  
In Section~\ref{sec:main}, we prove our main result: if the canonical form of a neural ideal is a Gr\"obner basis, then it is the universal Gr\"obner basis (Theorem~\ref{thm:main}).  
We also prove a partial converse: if the universal Gr\"obner basis of a neural ideal contains only so-called pseudo-monomials, then it is the canonical form (Theorem~\ref{thm:summary}).  
Our results motivate other questions: 
\begin{enumerate}
\item When is the canonical form a Gr\"obner basis?
\item If the universal Gr\"obner basis of a neural ideal is {\em not} a canonical form, what can the non-pseudo-monomial elements in the basis tell us about the receptive fields of the code?
\end{enumerate}
Sections~\ref{sec:gb} and~\ref{sec:new-RF} provide some partial answers these questions. Finally, a discussion is in Section~\ref{sec:discussion}.

\section{Background}\label{sec:background}
This section introduces neural ideals and related topics, which were first defined by Curto, Itskov, Veliz-Cuba, and Youngs~\cite{neural_ring}, and recalls some basics about Gr\"obner bases. We use the notation $[n]:=\{1,2,\dots, n\}$.

\subsection{Neural codes and receptive fields} \label{sec:code}
	A \dfn{neural code} (also known as a \textbf{combinatorial code}) on $n$ neurons is a set of binary firing patterns $C\subset \{0,1\}^n$, that is, a set of binary strings of neural activity. Note that neither timing nor rate of neural activity are recorded 
    in a neural code. 
		
	An element $c\in C$ of a neural code is a \dfn{codeword}.  Equivalently, a codeword is determined by the set of neurons that fire: 
\[	\supp(c):=\{i\in [n] \mid c_i=1\}\subseteq [n]~.\]
	Thus, the entire code is identified with a set of subsets of co-firing neurons: $\supp (C) =\{\supp (c) \mid c\in C\} \subseteq 2^{[n]}.$

In many areas of the brain, neurons are associated with \dfn{receptive fields} in a {\em stimulus space}. 
Of particular interest are the receptive fields of \textit{place cells}, which are neurons that fire in response to an animal's location. More specifically, each place cell is associated with a \dfn{place field}, a convex region of the animal's physical environment where the place cell has a high firing rate~\cite{Oke1}.
 The discovery of place cells and related neurons (grid cells and head direction cells) won neuroscientists John O'Keefe, May Britt Moser, and Edvard Moser the 2014 Nobel Prize in Physiology and Medicine.

Given a collection of 
sets $\mathcal{U} = \{U_1,...,U_n\}$ in a stimulus space $X$ (here $U_i$ is the receptive field of neuron $i$), the {\bf receptive field code}, denoted by $C(\mathcal{U})$, is: 
\[
C(\mathcal{U})~:=~
\left\{c\in \{0,1\}^n ~:~ \left( \bigcap_{i\in \supp(c)} U_i\right) \setminus \left( \bigcup_{j \notin \supp(c)} U_j \right) \neq \emptyset 
\right\}~.
\]
As mentioned earlier, we often identify this code with the corresponding set of subsets of $[n]$.
Also, we use the following convention for the empty intersection: $\bigcap_{i \in \emptyset} U_i := X$.

\begin{example} \label{ex:first-example}
Consider the 
sets $U_i$ in a stimulus space $X$ depicted in Figure~\ref{fig:U}.
The corresponding receptive field code is $C(\mathcal{U})=\{\emptyset, 1,123, 13, 3 \}$.    

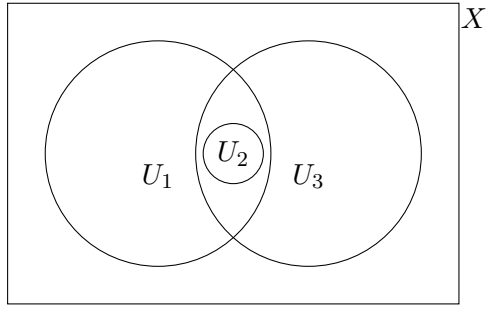
\begin{figure}
\begin{center}
            \def\firstcircle{(0,0) circle (1.5cm)}
            \def\secondcircle{(0:1cm) circle (.4cm)}
            \def\thirdcircle{(0:2cm) circle (1.5cm)}

        \begin{tikzpicture}
            \draw \firstcircle node[below] {$U_{1}$};
            \draw \secondcircle node {$U_{2}$};
            \draw \thirdcircle node [below] {$U_{3}$};
            \draw (-2,-2) rectangle (4,2);
            \node (X) at (4.2,1.8) {$X$};
        \end{tikzpicture}
    \end{center}
\caption{Receptive fields $U_i$ for which the code is $C(\mathcal{U})=\{\emptyset, 1,123, 13, 3 \}$. \label{fig:U}}
\end{figure}
\end{example}

\subsection{The neural ideal and its canonical form} \label{sec:CF}

A {\bf pseudo-monomial} in $\mathbb{F}_2[x_1,\dots, x_n]$ is a polynomial of the form 
\[
f~=~ \prod_{i\in \sigma} x_i\prod_{j\in\tau}(1+x_j)~,
\]
where $\sigma,\tau\subseteq[n]$ with $\sigma\cap\tau=\emptyset$.  Every term in a pseudo-monomial $f= \prod_{i\in \sigma} x_i\prod_{j\in\tau}(1+x_j)$ divides its highest-degree term, $\prod_{i\in \sigma \cup \tau} x_i$.  We will use this fact several times in this work.

Each $v\in \bifield$ defines a pseudo-monomial $\rho_v$ as follows:
\begin{align*}
\rho_v~:=~\prod_{i=1}^n(1-v_i-x_i)=\prod_{\{i \mid v_i=1\}}x_i\prod_{\{j \mid v_j=0\}}(1+x_j)=\prod_{\{i\in \supp(v)\}}x_i\prod_{\{j\not\in\supp(v)\}}(1-x_j)~.
\end{align*}
Notice that $\rho_v$ is the {\bf characteristic function} for $v$, that is, 
$\rho_v(x)=1$ if and only if $x=v$.  

\begin{definition}
Let $C \subseteq \{0,1\}^n$ be a neural code. The \textbf{neural ideal} $J_C$ is the ideal in $\mathbb{F}_2[x_1,\dots, x_n]$ generated by all $\rho_v$ for $v\not\in C$:
	\begin{align*}
    J_C~:=~\langle \{\rho_v|v\not\in C\}\rangle~.
    \end{align*}
\end{definition}

\noindent        
It follows that the variety of the neural ideal is the code itself: $V(J_C)=C$. The following lemma provides the algebraic version of the previous statement:

\begin{lemma}[Curto, Itskov, Veliz-Cuba, and Youngs~{\cite[Lemma 3.2]{neural_ring}}]\label{lemma_3.2}
Let $C \subseteq \{0,1\}^n$ be a neural code. Then 
\[	I(C) ~=~ J_C + \langle x_i(1 + x_i) \mid i \in [n] \rangle~, \]
where $I(C)$ is the ideal of the subset $C \subseteq \{0,1\}^n$.
\end{lemma}


Note that the ideal generated by the Boolean relations $\langle x_i(1 + x_i):i\in [n]\rangle$ is contained in $I(C)$, regardless of the structure of $C$.

A pseudo-monomial $f$ in an ideal $J$ in $\biring$ is {\bf minimal} if there does not exist another pseudo-monomial $g\in J$, with $g \neq f$,
such that $f=gh$ for some $h\in \biring$.

\begin{definition}
The \dfn{canonical form} of a neural ideal $J_C$, denoted by ${\rm CF}(J_C)$, 
is the set of all minimal pseudo-monomials of $J_C$. 
\end{definition}

Algorithms for computing the canonical form were given in~\cite{neural_ring,neural-hom,neural-ideal-sage}. In particular, \cite{neural-ideal-sage} describes an iterative method to compute the canonical form that is significantly more efficient than the original method presented in \cite{neural_ring}.

The canonical form ${\rm CF}(J_C)$ is a particular generating set for the neural ideal $J_C$~\cite{neural_ring}.  
The main goal in this work is to compare ${\rm CF}(J_C)$ to other generating sets of $J_C$, namely, its Gr\"obner bases.  

\begin{example} \label{ex:first-example-part-2}
Returning to Example~\ref{ex:first-example}, the codewords $v$ that are {\em not} in $C(\mathcal{U})= \{\emptyset, 1,123, 13, 3 \}$ are $2$, $12$, and $23$, so 
the neural ideal is 
	$J_C=\langle\{ x_2(1+x_1)(1+x_3),~x_1x_2(1+x_3),~x_2x_3(1+x_1)  \}\rangle $.
The canonical form is ${\rm CF}(J_{C(\mathcal{U})} )=\{ x_2(1+x_1),~x_2(1+x_3) \}$.  We will interpret these 
canonical-form polynomials in 
Example~\ref{ex:first-example-part-3} 
below.
\end{example}
    
\subsection{Receptive-field relationships} \label{sec:RF}

It turns out that we can interpret pseudo-monomials in $J_C$ (and thus in the canonical form) 
in terms of relationships among receptive fields. First we need the following notation: for any $\sigma\subseteq [n]$, define:
\[
x_\sigma ~:=~\prod_{i\in\sigma}x_i \quad 
\text{ and } \quad
U_\sigma ~:=~ \intersect_{i\in\sigma} U_i~,
\]
where, by convention, the empty intersection is the entire space $X$. 
\begin{lemma}[Curto, Itskov, Veliz-Cuba, and Youngs~{\cite[Lemma 4.2]{neural_ring}}] \label{lem:receptivefields}
	Let $X$ be a stimulus space, let $\mathcal{U}=\{U_i\}_{i=1}^n$ be a collection of 
    sets in $X$, and consider the receptive field code $C=C(\mathcal{U})$. Then for any pair of subsets $\sigma,\tau\subseteq [n]$, 
	$$x_\sigma\prod_{i\in\tau}(1+x_i)\in J_C \iff U_\sigma\subseteq\union_{i\in\tau}U_i~.$$
\end{lemma}
Thus, three types of receptive-field relationships (RF relationships) can be read off from pseudo-monomials in a neural ideal (e.g., those in the canonical form)~\cite{neural_ring}:
\begin{itemize}
\item[Type 1:] $x_\sigma \in J_C \iff U_\sigma = \emptyset $ (where $\sigma \neq \emptyset$).
\medskip
\item[Type 2:] $x_\sigma\prod_{i\in\tau}(1+ x_i)\in J_C \iff U_\sigma\subseteq\union_{i\in\tau}U_i$ (where $\sigma, \tau \neq \emptyset$).
\medskip
\item[Type 3:] $\prod_{i\in\tau}(1+x_i)\in J_C \iff X\subseteq\union_{i\in\tau}U_i$ (where $\tau \neq \emptyset$), and thus $X = \union_{i\in\tau}U_i$.
\end{itemize}

\begin{example} \label{ex:first-example-part-3}
The canonical form in Example~\ref{ex:first-example-part-2}, which is 
$\{ x_2(1+x_1),~x_2(1+x_3) \}$, encodes two Type 2 relationships: $U_2 \subseteq U_1$ and $U_2 \subseteq U_3$.  Indeed, we can verify this in Figure~\ref{fig:U}.
\end{example}

In this work, we reveal more types of RF relationships, which arise from non-pseudo-monomials.  They often appear in Gr\"obner bases of neural ideals (see Section \ref{sec:new-RF}).

\subsection{Gr\"obner bases} \label{sec:GB}
Here we recall some basics about Gr\"obner bases~\cite{gb-book,CLO-ideals,cocoa}.

Fix a monomial ordering $<$ of a polynomial ring $R=k[x_1,\dots, x_n]$ over a field $k$, and let $I$ be an ideal in $R$.  
Let $LT_<(I)$ denote the ideal generated by all leading terms, with respect to the monomial ordering $<$, of elements in $I$.

\begin{definition}
A \dfn{Gr\"obner basis} of $I$, with respect to $<$, is a finite subset 
of $I$ whose leading terms generate  $LT_<(I)$.  
\end{definition}
\noindent
One useful property of a \gb\space is that given a polynomial $f$ and a \gb\space $G$, the remainder of $f$ when divided by the set of elements in $G$ is uniquely determined. 

A Gr\"obner basis is \dfn{reduced} if (1) every $f \in G$
has leading coefficient 1, and (2) no term of any $f \in G$ is divisible by the leading term of any $g \in G$ for which $g\neq f$.   For a given monomial ordering, the reduced \gb\space of an ideal is unique. 

\begin{definition}
A \textbf{universal \gb} of an ideal $I$ is a  \gb\space that is a  \gb\space with respect to {\em every} monomial ordering. \textbf{The universal \gb} of an ideal $I$ is the union of all the reduced \gbs\space of $I$. 
\end{definition}

\textit{The} universal \gb\space is an instance of \textit{a} universal \gb, given that
the set of all distinct reduced \gbs\space of an ideal $I$ is finite~\cite[pg.\ 515]{gb-book}. This fact is actually the main result of the theory of \textit{Gr\"obner fans} first introduced in \cite{gb-fan}. 

\section{Main Result}\label{sec:main}
In this section, we give the main result of our paper: if the canonical form is a Gr\"obner basis, then it is the universal Gr\"obner basis (Theorem~\ref{thm:main}). 
Beyond being a natural expansion of some of Curto {\em et al.}'s results~\cite{neural_ring}, 
our theorem is also of mathematical interest since there are few classes of ideals whose universal Gr\"obner bases are known.  Indeed, such characterizations in general are known to be computationally difficult.  
\begin{thm} \label{thm:main}
If the canonical form of a neural ideal $J_C$ is a Gr\"obner basis of $J_C$ with respect to some monomial ordering, then it is \underline{the} universal Gr\"obner basis of $J_C$.
\end{thm}
The proof of Theorem~\ref{thm:main}, which appears in Section~\ref{sec:proof-main-theorem}, requires the following related results:
\begin{lemma}\label{lem:pseudo}
For a pseudo-monomial $f = x_{\sigma} \prod_{j \in \tau} (1+x_j)$ in $\mathbb{F}_2[x_1, \dots, x_n]$, the leading term of~$f$ with respect to {\em any} monomial ordering is its highest-degree term, $x_{\sigma\cup \tau}$.
\end{lemma}
\begin{proof}
This follows from the fact that every term of $f$ divides 
$x_{\sigma\cup \tau}$, and two properties of a monomial ordering~\cite{CLO-ideals}: it is a well-ordering (so, $1<x_i$), and $x_{\alpha} < x_{\beta}$ implies $x_{\alpha\cup \gamma} < x_{\beta \cup \gamma}$. 
\end{proof}

\begin{prop} \label{prop:a-universal}
If the canonical form of a neural ideal $J_C$ is a Gr\"obner basis of $J_C$ with respect to some monomial ordering, then it is \underline{a} universal Gr\"obner basis of $J_C$.
\end{prop}
\begin{proof}
Let $G$ denote the canonical form, and assume that $G$ is a Gr\"obner basis with respect to some monomial ordering $<_1$. 
Let $<_2$ denote another monomial ordering.
As always, we have the containment $\LT_{<_2}(G) \subseteq \LT_{<_2}(J_C)$, which we must prove is an equality.  
Accordingly, let $f \in J_C$.  
We must show that $\LT_{<_2}(f) \in \LT_{<_2}(G)$.
With respect to $<_1$, the reduction of $f$ by $G$ is 0, so we can write $f$ as a polynomial combination of some of the $g_i\in G$ in the following form:
\begin{align} \label{eq:f-sum-1}
f~=~
\frac{\LT_{<_1}(f)}{\LT(g_{1})}g_{1}+\frac{\LT_{<_1}(r_1)}{\LT(g_{2})}g_{2}+\dots+ 
	\frac{\LT_{<_1}(r_{t-1})}{\LT(g_{t})}g_{t}
~=~h_1+\dots+h_t~,
\end{align}
where (for $i=1,\dots, t$) we have  $g_i \in G$, $h_i := \frac{\LT_{<_1}(r_{i-1})}{\LT(g_{i})}g_{i}$, $r_0:= f$, and $r_i=f-h_1-\dots-h_i$ is the remainder after the $i$-th division of $f$ by $G$.  Note that in equation \eqref{eq:f-sum-1}, the polynomial $g_i$ may appear multiple times, but this does not affect our arguments.
By Lemma~\ref{lem:pseudo}, the leading term of $g_i$ does not depend on the monomial ordering. 
%
Moreover, each $h_i$ is the product of a monomial and a pseudo-monomial, $g_i$, so 
by a straightforward generalization of Lemma~\ref{lem:pseudo},
the leading term of $h_i$ with respect to {\em any} monomial ordering is $\LT_{<_1}(h_i)$. 
Also note that when dividing by the Gr\"obner basis $G$, 
$\LT_{<_1}(r_i) <_1 \LT_{<_1}(r_{i-1})$  so
the $\LT_{<_1}(r_{i})$ are distinct. 
This implies that the $\LT_{<_1}(h_{i})$ are distinct since $\LT_{<_1}(h_{i}) = \LT_{<_1}(r_{i-1})$.

Hence, among the list of monomials $\{\LT(h_{i})\}_{i=1}^t$, there is a unique largest monomial with respect to $<_2$, which we denote by $\LT(h_{i^*})$.  
Next, by examining the sum in~\eqref{eq:f-sum-1}, 
and noting that 
every term of $h_i$ divides the leading term of $h_i$, 
we see that $\LT_{<_2}(f)=\LT(h_{i^*})$.
Thus, because $g_{i^*}$ divides $h_{i^*}$, it follows that $\LT(g_{i^*})$ divides $\LT_{<_2}(f)$, and so, $\LT_{<_2}(f) \in \LT_{<_2}(G)$. 

Thus, if the canonical form is a Gr\"obner basis with respect to {\em some} monomial ordering, then it is a Gr\"obner basis with respect to {\em every} monomial ordering. 
\end{proof}

\subsection{Pseudo-monomials and hypercubes} \label{sec:hypercube}
To prove our main result (Theorem~\ref{thm:main}), we need to develop the connection between pseudo-monomials and hypercubes in the Boolean lattice.  The {\bf Boolean lattice} on $[n]$ is the power set $P([n]):=2^{[n]}$, partially ordered by inclusion.  Also,
for $\sigma \subseteq [n]$, we
let $P(\sigma)$ denote the power set of $\sigma$.

The {\bf support} of a monomial $\prod_{i=1}^n x_i^{a_i}$ is the set $\{ i \in [n] \mid a_i>0 \}$.

\begin{definition} \label{def:hypercube}
Let $f= x_{\sigma} \prod_{j \in \tau} (1+x_j)$ be a pseudo-monomial in $\mathbb{F}_2[x_1, \dots, x_n]$.  The {\bf hypercube of $f$}, denoted by $H(f)$, is the sublattice of the Boolean lattice on $[n]$ formed by the support of each term of $f$.
\end{definition}

\begin{remark} \label{rem:hypercube}
The hypercube of $f$ is the {\em interval} of the Boolean lattice from $\sigma$ to $\sigma \cup \tau$:
\[
H(f)=\{ \omega \mid \sigma \subseteq \omega \subseteq \sigma \cup \tau \} \subseteq P([n])~,
\]
and thus its Hasse diagram is a hypercube (this justifies its name).  
This is because:
\[ f=x_{\sigma} \prod\limits_{j \in \tau} (1+x_j) = \sum\limits_{\{ \theta \mid \theta \subseteq \tau\} } x_{\sigma \cup \theta}~.\]  
\end{remark}

\begin{example} \label{ex:hypercube}
Let $f=x_1x_2(1+x_3)(1+x_4)=x_1x_2x_3x_4+x_1x_2x_3+x_1x_2x_4+x_1x_2$. Figure~\ref{fig:hypercube} shows part of the Hasse diagram of $P([4])$, with the hypercube of $f$ indicated by circles and solid lines.  

\begin{figure}[ht]
	\begin{center}
		\begin{tikzpicture}
		
		\node[circle, draw] (1234) at (0,0) {$1234$};
		
		\node[circle, draw] (124) at (-1.25, .7*\h) {$124$};
		\node (134) at (1.25,.7*\h) {$134$};
		\node[circle, draw] (123) at (-3.75, .7*\h) {$123$};
		\node (234) at (3.75, .7*\h){$234$};
		
		\node[circle, draw] (12) at (-5, 2*\h) {$12$};
		\node (14) at (-3, 2*\h) {$13$}; 
		\node (13) at (-1, 2*\h) {$14$};
		\node (23) at (1, 2*\h) {$23$};
		\node (24) at (3, 2*\h) {$24$};
		\node (34) at (5, 2*\h) {$34$};
		
		\node (1) at (-3.8, 3*\h) {$1$};
		\node (2) at (-1.3, 3*\h) {$2$};
		\node (3) at (1.2, 3*\h) {$3$};
		\node (4) at (3.7, 3*\h) {$4$};
		
		\node (null) at (0, 3.7*\h) {$\emptyset$};
        
		\draw [red, very thick] (1234) -- (123) -- (12) -- (124) -- (1234);
		\draw [blue, loosely dashed](134) -- (14) -- (1) -- (13) -- (134);
		\draw [brown, loosely dashed] (234) -- (23) -- (2) -- (24) -- (234);
		\draw [gray, loosely dashed] (34) -- (3) -- (null) -- (4) -- (34) ;
		
		\draw [magenta, dotted, thick] (12) -- (1) -- (null) -- (2) -- (12);
		\end{tikzpicture}
	\end{center}
\caption{Displayed is part of the Hasse diagram of the Boolean lattice $P([4])$. The hypercube of $f=x_1x_2(1+x_3)(1+x_4)$ is indicated by circles and solid lines, and $P([2])$ is marked by dotted lines.  If $g$ is a pseudo-monomial that divides $f$, then its hypercube is contained in
either the hypercube of $f$ or 
one of the dashed-line squares ``parallel" to the hypercube of $f$
(see Example~\ref{ex:continued}). \label{fig:hypercube}}    
\end{figure}
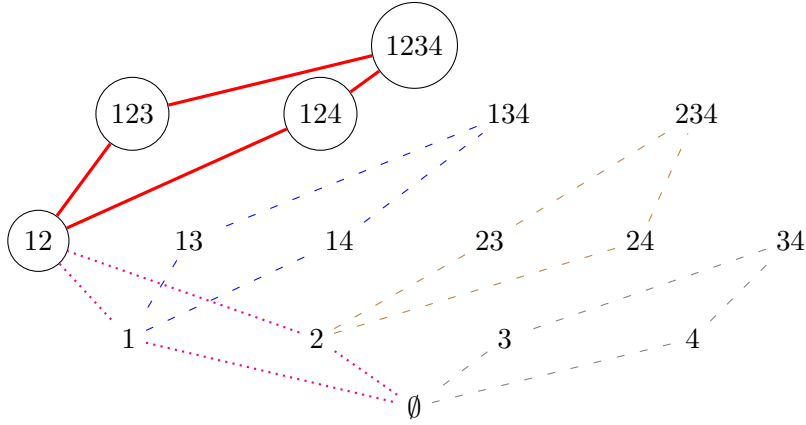    
    
\end{example}
Via hypercubes, divisibility of pseudo-monomials has a nice geometric interpretation:
\begin{lemma} \label{lem:hypdiv}
	For pseudo-monomials $f=\psmj{\sigma}{\tau}$ and $g=\psmj{\alpha}{\beta}$, the following are equivalent:
    \begin{enumerate}
	\item $g|f$,
	\item $\alpha\subseteq \sigma$ and $\beta \subseteq \tau$, 
    \item $H(g) \subseteq P(\sigma\cup\tau)$ and $ H(g) \cap P(\sigma)  = \{\alpha\}$, and 
    \item $H(g) \subseteq P(\sigma\cup\tau)$ and $ \left| H(g) \cap P(\sigma) \right| = 1 $.
	\end{enumerate}    
\end{lemma}

\begin{proof}
The implication (1) $\Leftarrow$ (2) is clear, and (1) $\Rightarrow$ (2) follows from the fact that $\biring$ is a unique factorization domain.  
For (2) $\Rightarrow$ (3), assume that $\alpha\subseteq \sigma$ and $\beta \subseteq \tau$.  Then $H(g) \subseteq P(\alpha \cup \beta) \subseteq P(\sigma \cup \tau)$.  So, we need only show that $H(g) \cap P(\sigma)  = \{ \alpha \}$.  To see this, we first recall:
\begin{align} \label{eq:hypercube-g}
H(g) ~=~ \{ \alpha \cup \theta \mid \theta \subseteq \beta \}
\end{align}
from Remark~\ref{rem:hypercube}. 
Thus,  
\begin{align*}
H(g) \cap P(\sigma)
~=~ 
 \{ \alpha \cup \theta \mid \theta \subseteq \beta \text{ and } \theta \subseteq \sigma \}
~=~ 
 \{\alpha\}~,
\end{align*}
where the second equality follows from hypotheses: $\alpha \subseteq \sigma$ and 
$\sigma \cap \beta \subseteq \sigma \cap \tau = \emptyset$ (because $\beta \subseteq \tau$).

(3) $\Rightarrow$ (4) is clear, so we need only show (2) $\Leftarrow$ (4).  Accordingly, suppose $H(g) \subseteq P(\sigma \cup \tau)$ and $I:=H(g) \cap P(\sigma)$ consists of only one element.  We claim that this element is $\alpha$.  Indeed, let $\omega \in I$ (i.e., $\omega \in H(g)$ and $\omega \subseteq \sigma$); then, $\alpha $ also is in $I$ (because $\alpha \in H(g)$ and $\alpha \subseteq \omega \subseteq \sigma$). So, $\alpha = \omega \subseteq \sigma$. 

To complete the proof, we must show that $\beta \subseteq \tau$.  To this end, let $k \in \beta$.  Then $\alpha \cup \{k\}$ is in $H(g)$, by equation~\eqref{eq:hypercube-g}, so it is {\em not} in $P(\sigma)$ (because $H(g) \cap P(\sigma)= \{\alpha\}$).  So, $k \in ( \beta \setminus \sigma )$.  Finally, $( \beta \setminus \sigma ) \subseteq \tau$, because $\alpha \cup \beta \subseteq \sigma \cup \tau$ follows from the hypothesis $H(g) \subseteq P(\sigma \cup \tau)$.  So, $k \in \tau$.  
\end{proof}

\begin{example} \label{ex:continued}
We return to the pseudo-monomial $f=x_1x_2(1+x_3)(1+x_4)$, which we rewrite as $f=x_{\sigma} \prod_{j \in \tau} (1+x_j)$, where $\sigma=\{1,2\}$ and $\tau=\{3,4\}$.  
In Figure~\ref{fig:hypercube}, $P(\sigma)=P([2])$ is marked by the dotted line. According to Lemma~\ref{lem:hypdiv}, a pseudo-monomial $h$ divides $f$ if and only if the hypercube of $h$ satisfies two conditions:
it includes a vertex from $P(\sigma)$, and
it is contained within either the hypercube of $f$ or 
one of the dashed-line squares ``parallel" to the hypercube of $f$ in Figure~\ref{fig:hypercube}.
\end{example}

\subsection{Multivariate division by pseudo-monomials} 
\label{sec:reduction}

The following result concerns reducing a given pseudo-monomial by a set of pseudo-monomials.
\begin{thm}\label{thm:divbyset}
	Consider a pseudo-monomial 
    $f=\psm{\sigma}{\tau}\in \biring$, and let $G$ be a finite set of pseudo-monomials in $\biring$. If some remainder upon division of $f$ by $G$ is $0$ for some monomial ordering, then there exists $g \in G$ such that $g$ divides $f$.
\end{thm}

\begin{proof}
Suppose that some remainder on division of $f$ by $G$ is 0:
\begin{align} \label{eq:f-sum}
f~=~
\frac{\LT(f)}{\LT(g_{1})}g_{1}+\frac{\LT(r_1)}{\LT(g_{2})}g_{2}+\dots+ 
	\frac{\LT(r_{t-1})}{\LT(g_{t})}g_{t}
~=~h_1+\dots+h_t~,
\end{align}
where, as in the proof of Proposition~\ref{prop:a-universal}, for $i=1,\dots, t$, we have $g_i \in G$, $h_i:= 	\frac{\LT(r_{i-1})}{\LT(g_{i})}g_{i}$, and $r_i=f-h_1-\dots-h_i$ is the remainder after the $i$-th division (and $r_0:= f$).  
Also, each term of $h_i$ divides the leading term of $h_i$.

By construction, $g_i | h_i$. So, it suffices to show that there exists $i$ such that $h_i | f$.  

We now claim that $\LT(h_i)|\LT(f)$ holds for all $i$.  We prove this claim by induction on $i$.  For the $i=1$ case, $\LT(h_1)=\LT(f)$.
If $i \geq 2$, then $\LT(h_i)$ is the leading term of:
\begin{align} \label{eq:r}
r_{i-1} = f-h_1-\dots - h_{i-1}~.
\end{align}
We now examine the summands in~\eqref{eq:r}.  
As $f$ is a pseudo-monomial, each term in $f$ divides $\LT(f)$, and the same holds for each remaining summand $h_i$: as noted above, its terms divide $\LT(h_i)$, and thus (by induction hypothesis) divide $\LT(f)$.  So, $\LT(h_i)= \LT(r_{i-1}) | \LT(f)$, proving our claim.

We now assert that $h_i$ is a pseudo-monomial. 
To see this, recall that $h_i$ is the product of a monomial and a pseudo-monomial (namely, $g_i$), so we just need to show that its leading term is square-free.  Indeed, this follows from two facts: $\LT(h_i)| \LT(f)$ and $f$ is a pseudo-monomial.


Hence, $H(h_i) \subseteq P(\sigma \cup \tau)$ for every $i$, because every term in $h_i$ divides $\LT(h_i)$ which in turn divides $x_{\sigma \cup \tau}=\LT(f)$.  
Thus, by Lemma~\ref{lem:hypdiv}, it is enough to show that
$\left| H(h_i) \cap P(\sigma) \right| = 1$  for some $i$ (because this would imply that $h_i|f$).

The sum in~\eqref{eq:f-sum} is over $\mathbb{F}_2$, so the polynomials $f, h_1, \dots, h_t$ together must contain an even number of each term.  We focus now on only those terms with support in $P(\sigma)$. The pseudo-monomial $f$ has only one such term (namely, $x_{\sigma}$).  Thus, 
some $h_{i^*}$ has an odd number of terms in $P(\sigma)$, i.e., $\left| H(h_{i^*}) \cap P(\sigma) \right|$ is odd.  On the other hand, both $H(h_{i^*})$ and $P(\sigma) $ are hypercubes in the Boolean lattice, so their intersection, if nonempty, also is a hypercube and thus has size $2^q$ for some $q \geq 0$.  Hence, $q=0$, so $\left| H(h_{i^*}) \cap P(\sigma) \right|=1$. This completes our proof.

\end{proof}

\subsection{Proof of Theorem~\ref{thm:main}} \label{sec:proof-main-theorem}
 Theorem~\ref{thm:divbyset} allows us to prove that when a canonical form is a Gr\"obner basis, it is reduced:
\begin{prop}\label{thm:reduced}
	If the canonical form of a neural ideal $J_C$ is a Gr\"{o}bner basis of $J_C$, then it is a reduced Gr\"{o}bner basis of $J_C$.
\end{prop}

\begin{proof}
	Suppose for contradiction that $\CF(J_C)$ is a Gr\"{o}bner basis, but not a reduced Gr\"{o}bner basis. 
Then there exist $f, g \in \CF(J_C)$, with $f \neq g$, such that $\LT(g)$ divides some term of $f$.  Thus, $\LT(g)$ divides $\LT(f)$ (because every term in a pseudo-monomial divides the leading term).  Thus, $\CF(J_C)$ and $\CF(J_C) \setminus \{f\}$ both generate the same ideal of leading terms, and hence $\CF(J_C) \setminus \{f\}$ is also a Gr\"obner basis of $J_C$.  It follows that the remainder on division of $f$ by $\CF(J_C) \setminus \{f\}$ is 0, so by Theorem~\ref{thm:divbyset}, there exists $h \in \CF(J_C) \setminus \{f\}$ such that $h | f$.  Hence, $f$ is a non-minimal element of the canonical form, which is a contradiction.
\end{proof}

Now we can prove Theorem~\ref{thm:main}, which states that a canonical form that is a Gr\"obner basis is the universal Gr\"obner basis:
\begin{proof}[Proof of Theorem~\ref{thm:main}]
Follows from Propositions \ref{prop:a-universal} and \ref{thm:reduced}.
\end{proof}

\subsection{Every pseudo-monomial in a reduced Gr\"obner basis is in the canonical form}
In this subsection, we prove the following partial converse of Theorem~\ref{thm:main}: if the universal Gr\"obner basis of a neural ideal consists of only pseudo-monomials, then it equals the canonical form (Theorem~\ref{thm:summary}). 

We first show that every pseudo-monomial in a reduced Gr\"obner basis is in the canonical form.
\begin{prop} \label{prop:pseudo-mon-in-reduced} Let $J_C$ be a neural ideal.
\begin{enumerate}
\item Let $G$ be a reduced Gr\"obner basis of $J_C$.  Then every pseudo-monomial in $G$ is in the canonical form of $J_C$.
\item Let $\widehat{G}$ be the universal Gr\"obner basis of $J_C$.  Then every pseudo-monomial in $\widehat{G}$ is in the canonical form of $J_C$.
\end{enumerate}
\end{prop}
\begin{proof}
Let $f$ be a pseudo-monomial in $G$.  Suppose that $f$ is {\em not} a minimal pseudo-monomial in $J_C$:
for some pseudo-monomial $h\in J_C$ such that deg($h$)\textless deg($f$), $h|f$. 
Then for some $g\in G$, $\LT(g)|\LT(h)$. Hence, $\LT(g) | \LT(f)$ (because $\LT(h) | \LT(f)$) and also $g\neq f$ (because $\deg(g) \leq \deg(h) < \deg(f)$). 
This is a contradiction: $f$ and $g$ cannot both be in a reduced Gr\"{o}bner basis. 

Finally, (2) follows directly from (1).
\end{proof}


\begin{thm} \label{thm:summary}
Let $J_C$ be a neural ideal.  The following are equivalent:
\begin{enumerate}
\item the canonical form of $J_C$ is a Gr\"obner basis of $J_C$,
\item the canonical form of $J_C$ is {the} universal Gr\"obner basis  of $J_C$, and
\item {the} universal Gr\"obner basis of $J_C$ consists of pseudo-monomials.
\end{enumerate}
\end{thm}
\begin{proof}
The implication (1)$\Rightarrow$(2) is Theorem~\ref{thm:main}, and both (1)$\Leftarrow$(2) and (2)$\Rightarrow$(3) are clear.  For (3)$\Rightarrow$(1), assume that the universal Gr\"obner basis $\widehat{G}$ consists of pseudo-monomials.  Then, by Proposition~\ref{prop:pseudo-mon-in-reduced}(2), $\widehat{G}$ is contained in the canonical form of $J_C$.  Thus, the canonical form contains a Gr\"obner basis of $J_C$ (namely, $\widehat{G}$) and hence is itself a Gr\"obner basis. 
\end{proof}

\begin{remark} \label{rem:universal}
Suppose we want to know whether a code's canonical form is a Gr\"obner basis.  
Theorem~\ref{thm:summary} tells us how to do so {\em without} computing the canonical form: compute the universal Gr\"obner basis, and then check whether it contains only pseudo-monomials.  See Example~\ref{ex:universal-code}.

Under certain conditions, e.g. small number of neurons, 
computing the Gr\"obner basis is more efficient than computing the canonical form,
but 
is there some way to avoid computations entirely and yet still decide whether the canonical form is a Gr\"obner basis?
In the next section, we give conditions under which we can resolve this decision problem quickly.
\end{remark}

\begin{example} \label{ex:universal-code}
Consider the neural code $C = \{0100,0101,0111\}$. 
The universal Gr\"obner basis of $J_C$ is 
$\widehat{G} = \{ x_3(x_4 + 1) ,~ x_2 + 1,~ x_1\}$, 
so it
contains only 
pseudo-monomials.  
Thus, by Theorem~\ref{thm:summary}, 
$\widehat{G}$ is the canonical form. 
\end{example}

\begin{example} \label{ex:not-universal-code}
Consider the neural code $C = \{0101,1100,1110\}$.  
The universal Gr\"obner basis of $J_C$ is 
$\widehat{G} = \{x_4 x_3,~ x_3  (x_1 + 1),~ x_1 + x_4 + 1,~ x_2 + 1\}$,
which contains the non-pseudo-monomial $x_1+x_4+1$.  
Thus, by Theorem \ref{thm:summary}, the canonical form is not a universal Gr\"obner basis of $J_C$.
Indeed, the canonical form is 
${\rm CF}(J_C)=\{x_3  (x_1 + 1),~ x_2 + 1,~ (x_4 + 1)  (x_1 + 1),~ x_4  x_1,~ x_4 x_3\}$, and, 
for a monomial ordering where $x_4>x_1$, the leading term 
of 
the non-pseudo-monomial $x_1+x_4+1$ is $x_4$, which 
is {\em not} divisible by any of the leading terms from the canonical form.  
\end{example}

\section{When is the canonical form a Gr\"obner basis?}\label{sec:gb}


In this section we present some results that partially solve the question of when is the canonical form a Gr\"obner basis for the neural ideal. A complete answer to this question is not only of theoretical interest but perhaps also of practical relevance. Extensive computations suggest that, under certain conditions, Gr\"obner bases of neural ideals can be computed more efficiently than canonical forms. This is true for small neural codes.
Moreover, the iterative nature of the newer canonical form algorithm hints towards the ability to compute canonical forms and Gr\"obner bases of neural codes in large dimensions by `gluing' those of codes on small dimensions. Such decomposition results are a common theme in other areas of applied algebraic geometry 
such as algebraic statistics and phylogenetic algebraic geometry 
\cite{AR08, EKS14}.



Table \ref{runtimes} displays a runtime comparison between the iterative canonical form algorithm described in \cite{neural-ideal-sage} and a specialized Gr\"obner basis algorithm for Boolean rings implemented in SageMath based on the work in \cite{polybori}. We report the mean time (in seconds) 
of 100 randomly generated codes on $n$ neurons
for 
$n = 4, \dots, 8$. More precisely, for each code, a number $m$ was chosen uniformly at random from $\{1,\dots, 2^n-1 \}$ and then $m$ codewords were chosen at random.
These computations were performed on SageMath 7.2 running on a Macbook Pro with a 2.8 GHz Intel Core i7 processor and 16 GB of memory. 

\begin{table}[!hbtp]
\centering
\begin{tabular}{|l|l|l|l|l|l|} 
\hline
Dimension & 4 & 5 & 6 & 7 & 8 \\ \hline 
Canonical form & 0.0016 & 0.0076 & 0.108 & 0.621 & 1.964 \\ 
\hline
Gr\"obner basis & 0.00147 & 0.00202 & 0.00496 & 0.01604 & 0.16638  \\ \hline
\end{tabular}
    \caption{Runtime comparison of canonical form versus Gr\"obner basis computations.}
\label{runtimes}
\end{table}

For codes on a larger number of neurons, our computations indicate that in general Gr\"obner bases computations are still more efficient than canonical form computations. However, even in the case of $n=9$ neurons we found codes whose Gr\"obner bases took over 6 hours to be computed.


\begin{prop} \label{prop:size-of-C}
Let $C$  be a neural code on $n$ neurons. If $|C| = 1$ or $|C| = 2^n-1$, then the canonical form of $J_C$ is the universal Gr\"obner basis of $J_C$. 
\end{prop}

\begin{proof}
If $C = \{c\}$, then Lemma \ref{lemma_3.2} implies that $J_C = \langle x_1-c_1,~ x_2-c_2,~\dots,~ x_n-c_n\rangle$. When $|C| = 2^n-1$, then by definition $J_C = \langle \rho_v \rangle$ for the unique $v\notin C$. In either case, the indicated generating set  is both the canonical form and the universal Gr\"obner basis of $J_C$.
\end{proof}

A set of subsets $\Delta\subseteq 2^{[n]}$ is an (abstract) \textbf{simplicial complex} if $\sigma \in \Delta$ and $\tau \subseteq \sigma$ implies $\tau \in \Delta$. A neural code $C$ is a simplicial complex if its support $\mathrm{supp}(C)$ is a simplicial complex.   

\begin{prop}
If  $C$ is a simplicial complex, then the canonical form of $J_C$ 
is the universal Gr\"obner basis of $J_C$.
\end{prop}

\begin{proof}
If $C$ is a simplicial complex, then $J_C$ is a monomial ideal generated by the minimal Type~1 relationships (indeed, it is the Stanley-Reisner ideal of the simplicial complex $\mathrm{supp}(C)$)~\cite[Lemma 4.4]{neural_ring}.  These minimal Type-1 relationships comprise the canonical form of $J_C$, and also form the universal Gr\"obner basis of $J_C$.
%
\end{proof}

The next result gives 
conditions that guarantee that the canonical form is not a 
Gr\"obner basis. 

\begin{prop}
Let $\mathcal{U}=\{U_i\}_{i=1}^n$ be a collection of 
    sets in a stimulus space~$X$, and 
	let $C=C(\mathcal{U})$ denote the corresponding receptive field code. 
    If one of the following conditions hold, then the canonical form of $J_C$ is {\em not} a Gr\"obner basis of $J_C$:
\begin{enumerate}
\item[(1)] Two proper, nonempty receptive fields coincide:
$\emptyset \neq U_i = U_j \subsetneq X$ for some $i\neq j \in [n]$. 
\item[(2)] Two nonempty receptive fields are complementary:
$U_i = X \setminus U_j$ for some $i\neq j \in [n]$ with $U_i \neq \emptyset$ and $U_j \neq \emptyset$.
\end{enumerate}
\end{prop}

\begin{proof} (1)
Suppose $U_i, U_j \in \mathcal{U}$ are two sets with $\emptyset \neq U_i = U_j \subsetneq X$. 
By Lemma \ref{lem:receptivefields}, both $f = x_i(x_j+1)$ and $g = x_j(x_i+1)$ are in $J_C$. In fact, $f$ and $g$ are minimal pseudo-monomials in $J_C$ (because $\emptyset \neq U_i=U_j \neq X$), so $f, g \in \mathrm{CF}(J_C)$. Under any monomial ordering, $\mathrm{LT}(f) = \mathrm{LT}(g) = x_ix_j$ (by Lemma~\ref{lem:pseudo}), so the set $\mathrm{CF}(J_C)$ is not reduced and thus cannot be a reduced Gr\"obner basis. 
Hence, by Proposition \ref{thm:reduced}, $\mathrm{CF}(J_C)$ cannot be a Gr\"obner basis.

(2) 
Now assume that $U_i = X \setminus U_j$ for some $i\neq j \in [n]$, 
with $U_i \neq \emptyset$ and $U_j \neq \emptyset$.
Thus, $U_i\cap U_j = \emptyset$ and $U_i \cup U_j = X$, so Lemma \ref{lem:receptivefields} implies that $f = x_ix_j$ and $g = (x_i+1)(x_j+1)$ are
in $J_C$.  Now we proceed as in the previous paragraph: $f$ and $g$ are minimal 
pseudo-monomials in $\mathrm{CF}(J_C)$, and  $\mathrm{LT}(f) = \mathrm{LT}(g) = x_ix_j$, 
so, by Proposition \ref{thm:reduced}, $\mathrm{CF}(J_C)$ cannot be a Gr\"obner basis.
%
%
\end{proof}

The last result in this section concerns a class of codes that we call \textbf{complement-complete}. 

\begin{definition}
The \textbf{complement} of 
$c \in \{0,1\}^n$ is the codeword $\overline{c} \in \{0,1\}^n$ defined by $\overline{c}_i = 1$ if and only if $c_i = 0$. 
A neural code $C$ 
is \textbf{complement-complete} if for all $c \in C$, then $\overline{c}$ is also in $C$.
\end{definition}

\begin{example}\label{ex:comp-compcode}
The complement of the codeword $c_1=1000$ is $\overline{c_1}=0111$, and the complement of $c_2=1010$ is $\overline{c_2}=0101$.  Thus, the code $C=\{1000,0111,1010,0101\}$ is complement-complete.
\end{example}

\begin{definition}
The \textbf{complement} of a pseudo-monomial $f = x_{\sigma}\prod_{i\in \tau}(1+x_i)$ 
is the pseudo-monomial $\overline{f} = x_\tau\prod_{j\in \sigma}(1+x_j)$.
\end{definition}

\begin{lemma} \label{lem:bar}
Consider pseudo-monomials $f = x_{\sigma}\prod_{i\in \tau}(1+x_i)$ and $g = x_{\sigma'}\prod_{i\in \tau'}(1+x_i)$. If $f$ divides $g$, then $\overline{f}$ divides $\overline{g}$.
\end{lemma}

\begin{proof}
This follows from the fact that
$f\mid g$ if and only if $\sigma'\subseteq \sigma$ and $\tau' \subseteq \tau$ (Lemma \ref{lem:hypdiv}). 
\end{proof}

\begin{prop} \label{prop:complement-complete}
Let $C$ be a code on $n$ neurons, with $C \subsetneq \{0,1\}^n$. If $C$ is complement-complete, then the canonical form of $J_C$ is {\em not} a Gr\"obner basis of $J_C$.  
\end{prop}

\begin{proof}
Note that since $C \neq \{0,1\}^n$, $J_C$ is not trivial. We make the following claim: \\
{\sc Claim:} If $h$ is a pseudo-monomial in $J_C$, then $\overline{h}$ is also in $ J_C$.\\
To see this, 
let $S$ be the set of all degree-$n$ pseudo-monomials in $\mathbb{F}_2[x_1,\dots, x_n]$ that are multiples of $h$ (so, $S \subseteq J_C$). Degree-$n$ pseudo-monomials in $\mathbb{F}_2[x_1,\dots, x_n]$ are characteristic functions $\rho_v$, so, every element of $S$ is some $\rho_v$, where $v \notin C$.  Thus, every element of $\overline{S}:=\{\overline f \mid f \in S \}$ has the form $ \overline{\rho_v} = \rho_{\overline{v}} $, where $v \notin C$, which is equivalent to $\overline{v} \notin C$, as $C$ is complement-complete.  So, $\overline{S} \subseteq J_C$.  

Next, let $s \in S$, that is, $s = hq$ for some pseudo-monomial $q$.
Then $h\overline{q}$ is also in $S$. Since $\gcd(q,\overline{q}) = 1$, it follows that $h = \gcd(hq, h \overline{q})$, so $h = \gcd\{S\}$. Thus, $\overline{h} = \gcd\{\overline{S}\}$, so $\overline{h} \in J_C$
(because $\overline{S} \subseteq J_C$), which proves the claim.

Now let $f \in CF(J_C)$.
By the claim, $\overline{f}$ is in $J_C$, and now we assert that, like $f$, the pseudo-monomial $\overline{f}$ is in $CF(J_C)$. 
Indeed, if a pseudo-monomial $d$ in $J_C$ divides $\overline{f}$, then by Lemma~\ref{lem:bar}, the pseudo-monomial $\overline{d}$ divides $f$.  
Also, $\overline{d} \in J_C$ (by the claim), so 
$\overline{d}=f$ (because $f$ is minimal), and thus $d = \overline{f}$.  Hence, $\overline{f}$ is minimal, and so  $\overline{f}$ is also in $CF(J_C)$.  
Thus, $CF(J_C)$ contains two polynomials ($f$ and $\overline{f}$) with the same leading term, and so 
 is not a reduced Gr\"obner basis, and thus (by Proposition~\ref{thm:reduced}) is not a Gr\"obner basis of $J_C$. 
%
%
%
\end{proof}

\begin{example}  
Consider again the complement-complete code $C=\{1000,0111,1010,0101\}$ from Example \ref{ex:comp-compcode}.  The canonical form is $CF(J_C)=\{(x_1+1)(x_2+1),~(x_1+1)(x_4+1),~x_1x_2,~x_2(x_4+1),~x_1x_4,~x_4(x_2+1)\}$.  Note that 
$CF(J_C)$ is itself complement-complete;
for example, $f=x_2(x_4+1)$ and $\overline{f}=x_4(x_2+1)$ are both in $CF(J_C)$.  Also, we can show directly that $CF(J_C)$ is not a Gr\"obner basis, which is consistent with Proposition~\ref{prop:complement-complete}: with respect to any monomial ordering, the leading term of $f+\overline{f}=x_2+x_4$ 
is not divisible by any of the leading terms in $CF(J_C)$. 
\end{example}

\section{New receptive-field relationships} 
\label{sec:new-RF}

We saw earlier that if the universal Gr\"obner basis of a neural ideal consists of only pseudo-monomials, then it equals the canonical form (Theorem~\ref{thm:summary}).  When this is not the case, there are non-pseudo-monomial elements in the universal Gr\"obner basis, so it is natural to ask what they tell us about the receptive fields of the code.  In other words, what types of RF relationships, besides those of Types 1--3 (Lemma~\ref{lem:receptivefields}), appear in Gr\"obner bases?  Here we give a partial answer:
\begin{thm} \label{thm:new-RF-relns}
	Let  $\mathcal{U}=\{U_i\}_{i=1}^n$ be a collection of 
    sets in a stimulus space~$X$.
	Let $C=C(\mathcal{U})$ denote the corresponding receptive field code, and let $J_C$ denote the neural ideal.
    Then for any subsets $\sigma_1,\sigma_2,\tau_1,\tau_2 \subseteq [n]$, and $m$ indices $1 \leq i_1 < i_2 < \dots < i_m \leq n$, 
    with $m \geq 2$,  
    we have RF relationships as follows:
\begin{itemize}
\item[Type 4:]
$x_{\sigma_1} \prod_{i \in \tau_1} (1 + x_i) + x_{\sigma_2} \prod_{j \in \tau_2} (1 + x_j) \in J_C$
~$\Rightarrow$~
$U_{\sigma_1}\cap \left( \bigcap_{i\in \tau_1}U_i^c \right) = U_{\sigma_2}\cap \left( \bigcap_{j\in \tau_2}U_j^c\right)$.
\medskip
\item[Type 5:] 
$x_{i_1}+ \dots +x_{i_m}  \in J_C$
~$\Rightarrow$~
$U_{i_k}\subseteq\bigcup_{j\in [m] \setminus \{k\} }U_{i_j}$ for all $k=1, \dots, m$,
and if, additionally, $m$ is odd, then  $\bigcap_{k=1}^m U_{i_k}=\emptyset$.
\medskip
\item[Type 6:] 
$x_{i_1}+ \dots +x_{i_m}+1  \in J_C$ 
~$\Rightarrow$~
$\bigcup_{k=1}^m U_{i_k}=X$.
\end{itemize}
\end{thm}

\begin{proof}  Throughout the proof, for $p \in X$, we let $c(p)$ 
denote the corresponding codeword in $C$.

{\sc Type 4.} 
Let $f_1:=x_{\sigma_1} \prod_{i \in \tau_1} (1+x_i)$, and let $f_2:=x_{\sigma_2} \prod_{j \in \tau_2} (1+x_j)$.  
Also, let 
$W_1 := U_{\sigma_1} \cap \left( \bigcap_{i\in \tau_1}U_i^c \right)$, and 
let $W_2 := U_{\sigma_2} \cap \left( \bigcap_{j\in \tau_2}U_j^c \right)$. 
By symmetry, we need only show that $W_1 \subseteq W_2$.  
To this end, let $p \in W_1$ (so, $c(p) \in C$).  First, because $f_1+f_2 \in J_C$ and $V(J_C)=C$, it follows that $f_1(c(p))=f_2(c(p))$.  Next, for $i=1,2$, we have $p \in W_i$ if and only if $f_i(c(p))=1$.  Thus, $p \in W_2$.

{\sc Type 5.} 
Let $g:= x_{i_1}+\dotsm +x_{i_m}$.  By symmetry, we need only show that $U_{i_1} \subseteq \bigcup_{l=2}^m U_{i_l}$.  
To this end, let $ p \in U_{i_1}$ (so, $c(p)_{i_1}=1$).  Then $g \in J_C$ implies the following equality in $\mathbb{F}_2$: 
\begin{align} \label{eq:g}
 0 ~=~ g(c(p)) ~=~ c(p)_{i_1} + c(p)_{i_2} + \dots + c(p)_{i_m} ~=~  1 + c(p)_{i_2} + \dots + c(p)_{i_m}~.
\end{align}
Thus, for some $k \geq 2$, we have $c(p)_{i_k}=1$, i.e., $p \in U_{i_k}$.  Hence, $p \in \bigcup_{l=2}^m  U_{i_l}$.

Now assume, additionally, that $m$ is odd.  Suppose, for contradiction, that there exists $q \in \bigcap_{k=1}^m U_{i_k}$.  Then, like the sum~\eqref{eq:g} above, we have $0=g(c(q))=1+\dots+1=m$, which contradicts the hypothesis that $m$ is odd.  So, $\bigcap_{k=1}^m U_{i_k} = \emptyset$.

{\sc Type 6.} 
Let $h:=x_{i_1}+ \dots +x_{i_m}+1$.  Let $p \in X$ (so, $c(p) \in C$).  We must show that $p \in \bigcup_{k=1}^m U_{i_k}$.  Because $h \in J_C$, we have 
$0 = h(c(p)) = c(p)_{i_1} + \dots + c(p)_{i_m} + 1$.  Thus, for some $k \in [m]$, we have $c(p)_{i_k}=1$, i.e., $p \in U_{i_k}$.  Hence, $p \in \bigcup_{k=1}^m  U_{i_k}$.
\end{proof}

\begin{remark} \label{rem:equality}
Like the earlier RF relationships (those of Types 1--3 from Lemma~\ref{lem:receptivefields}),
some of our new ones (Types 4--6) are containments and some are equalities.
\end{remark}

\begin{example} \label{ex:type-6}
Recall the code $C = \{0101,1100,1110\}$, 
from Example~\ref{ex:not-universal-code}, for which the universal Gr\"obner basis of $J_C$ is 
$\widehat{G} = \{x_4 x_3,~ x_3  (x_1 + 1),~ x_1 + x_4 + 1,~ x_2 + 1\}$.  The polynomial $x_1+x_4+1$ encodes a Type 6 relationship: $U_1\cup U_4=X$. Also, the polynomial $x_2+1$ encodes a Type 3 relationship: $U_2=X$, which together gives us $U_1\cup U_4=U_2$.  The canonical form also contains the polynomial $x_1x_4$, which encodes a Type 1 relationship: $U_1\cap U_4=\emptyset$.  We conclude that 
$U_1 \dot \cup U_4=U_2$.

\end{example}

\begin{example} \label{ex:types-4-5}
Consider the code $C=\{00,11\}$.  The universal Gr\"obner basis of $C$ is $\widehat{G} = \{ x_1(1+x_1),~ x_1+x_2, ~x_2(1+x_2) \}$.  The polynomial $x_1+x_2$ encodes a Type 4 relationship: $U_1=U_2$.   (The polynomial $x_1+x_2$ also encodes Type 5 relationships.)
This points to one of the advantages of our new RF relationships:
we can read off some set equalities more quickly than from the canonical form.  Indeed, the canonical form is ${\rm CF}(J_C)=\{x_1(1+x_2),~x_2(1+x_1) \}$, in which the Type 2 relationships are $U_1 \subseteq U_2$ and $U_2 \subseteq U_1$ -- and only from there do we infer the equality  $U_1=U_2$.
\end{example}

\section{Discussion} \label{sec:discussion}
In this work, we proved that if a code's canonical form is a Gr\"obner basis of the neural ideal, then it is the universal Gr\"obner basis.  Additionally, we gave conditions that guarantee or preclude this situation, and found three new types of receptive-field relationships that arise in 
neural ideals.  Going forward, there are natural extensions to pursue:
\begin{enumerate}
\item Give a complete characterization of codes for which the canonical form is a Gr\"obner basis.
\item Beyond those of Types 1--6, what other receptive-field relationships can be read off from a Gr\"obner basis, 
	and what do they tell us about a code?
\end{enumerate}
Solutions to these problems would help us extract information about the receptive-field structure directly from the neural code.

Finally, we expect that our results can be used to improve canonical-form algorithms. 
Indeed, our experiments indicate that under certain conditions, Gr\"obner bases can be computed more efficiently than canonical forms.
Moreover, every pseudo-monomial in the universal 
Gr\"obner basis of a neural ideal is in the canonical form -- so, that subset of the canonical form can be obtained directly from the universal Gr\"obner basis.  And, in the case when the universal Gr\"obner basis contains only pseudo-monomials, then we can conclude immediately that the basis is in fact the canonical form. 
Moreover, we hope to develop decomposition results to build canonical forms and Gr\"obner basis of codes in large dimensions by `gluing' those of codes in smaller dimensions.

\subsection*{Acknowledgments}
{ 
DM, RK, and EP conducted this research as part of the 2015 Pacific Undergraduate Research Experience in Mathematics Interns Program funded by the NSF (DMS-1045147 and DMS-1045082) and the NSA (H98230-14-1- 0131), in which RG and LG served as mentors and KP was a GTA.   
JL conducted this research as part of the 2016 NSF-funded REU in the Department of Mathematics at Texas A\&M University (DMS-1460766), in which AS served as mentor and KP was a GTA.  The authors thank Ihmar Aldana, Carina Curto, Vladimir Itskov, and Ola Sobieska for helpful discussions.  LG was supported by the Simons Foundation Collaboration grant 282241.
AS was supported by the NSF (DMS-1312473/DMS-1513364). The authors thank an anonymous referee for helpful comments which improved this work.
}


\bibliographystyle{plain}
\bibliography{neuro}

\end{document}